\documentclass[a4paper,12pt]{llncs}
\usepackage{amsmath,amssymb}
\usepackage[left=4.17cm,top=4.17cm,right=4.17cm,bottom=4.17cm]{geometry}
\usepackage{epsfig}

\usepackage{color}

\usepackage{graphicx}

\usepackage{wrapfig}
\usepackage{amssymb}
\usepackage{amsmath}

\usepackage{longtable}
\usepackage{multirow}
\usepackage{array}
\usepackage{paralist}

\date{}
\begin{document}

\title{\Large{The Iterated Local Transitivity Model for Tournaments}\thanks{Research supported by a grant from NSERC.}}
\author{Anthony Bonato\inst{1}, Ketan Chaudhary\inst{1}}
\institute{Toronto Metropolitan University}

\maketitle

\begin{abstract}

A key generative principle within social and other complex networks is transitivity, where
friends of friends are more likely friends. We propose a new model for highly dense complex networks based on transitivity, called the Iterated Local Transitivity Tournament (or ILTT) model. In ILTT and a dual version of the model, we iteratively apply the
principle of transitivity to form new tournaments. The resulting models generate tournaments with small average distances as observed in real-world complex networks. We explore properties of small subtournaments or motifs in the ILTT model and study its graph-theoretic properties, such as Hamilton cycles, spectral properties, and domination numbers. We finish with a set of open problems and the next steps for the ILTT model.

\end{abstract}

\section{Introduction}

The vast volume of data mined from the web and other networks from the physical, natural, and social sciences, suggests a view of many real-world networks as self-organizing phenomena satisfying common properties. Such complex networks capture interactions in many phenomena, ranging from friendship ties in Facebook, to Bitcoin transactions, to interactions between proteins in living cells. Complex networks evolve via several mechanisms such as preferential attachment or copying that predict how links between nodes are formed over time. Key empirically observed properties of complex networks such as the small world property (which predicts small distances between typical pairs of nodes and high local clustering) have been successfully captured by models such as preferential attachment \cite{ba,bol}. See the book \cite{bbook} for a survey of early complex network models, along with \cite{at}.

\emph{Balance theory} cites mechanisms to complete triads (that is, subgraphs consisting of three nodes) in social and other complex networks \cite{ek,he}. A central mechanism in balance theory is \emph{transitivity}: if $x$ is a friend of $y,$ and $y$ is a friend of $z,$ then $x$ is a friend of $z$; see, for example, \cite{scott}. Directed networks of ratings or trust scores and models for their propagation were first considered in \cite{guha}. \emph{Status theory} for directed networks, first introduced in \cite{lhk}, was motivated by both trust propagation and balance theory. While balance theory focuses on likes and dislikes, status theory posits that a directed link indicates that the creator of the link views the recipient as having higher status. For example, on Twitter or other social media, a directed link captures one user following another, and the person they follow may be of higher social status. Evidence for status theory was found in directed networks derived from Epinions, Slashdot, and Wikipedia \cite{lhk}. For other applications of status theory and directed triads in social networks, see also \cite{shc,sm}.

The \emph{Iterated Local Transitivity} (\emph{ILT}) model introduced in \cite{ilt1,ilt} and further studied in \cite{ilm,ilat,mason}, simulates structural properties in complex networks emerging from transitivity. Transitivity gives rise to the notion of \emph{cloning}, where a new node $x$ is adjacent to all of the neighbors of some existing node $y$. Note that in the ILT model, the nodes have local influence within their neighbor sets. The ILT model simulates many properties of social networks. For example, as shown in \cite{ilt}, graphs generated by the model densify over time and exhibit bad spectral expansion. In addition, the ILT model generates graphs with the small-world property, which requires graphs to have low diameter and high clustering coefficient compared to random graphs with the same number of nodes and expected average degree. A directed analogue of the ILT model was introduced in \cite{directed}.

Tournaments are directed graphs where each pair of nodes shares exactly one directed edge (or arc). Tournaments are simplified representations of highly interconnected structures in networks. For example, we may consider users on Twitter or TikTok that follow each other or those on Reddit focused on a particular topic or community. A tournament arises from such social networks by assigning an arc  $(u,v)$ if the user $u$ responds more frequently to the posts of $v$ than $v$ does to $u$. Other examples of tournaments in real-world networks include those in sports, with edges corresponding to one player or team winning over another. Such directed cliques are of interest in network science as one type of \emph{motif}, which are certain small-order significant subgraphs. Tournaments in social and other networks grow organically; therefore, it is natural to consider models simulating their evolution.

The present paper explores a version of the ILT model for tournaments. The \emph{Iterated Local Transitivity for Tournaments (ILTT)} model is deterministically defined over discrete time-steps as follows. The only parameter of this deterministic model is the initial tournament $G=G_{0}$, which is called the \emph{base}. For a non-negative integer $t$, $G_{t}$ represents the tournament at time-step $t$. Suppose that the tournament $G_{t}$ has been defined for a fixed time-step $t\geq 0$. To form $G_{t+1}$, for each $x \in V(G_{t}$), add a new node $x'$ called the \emph{clone} of $x$. We refer to $x$ as the \emph{parent} of $x',$ and $x'$ as the \emph{child} of $x.$ We add the arc $(x',x)$, and for each arc $(x,y)$ in $G_t$, we add the arc $(x',y')$ to $G_{t+1}.$ For arcs $(x,z)$ and $(y,x)$ in $G_{t}$, we add arcs $(x',z)$ and $(y,x'),$ respectively, in $G_{t+1}$. See Figure~1. We refer to $G_t$ as an \emph{ILTT tournament}. Note that the subtournament of clones in $G_{t+1}$ is isomorphic to $G_t;$ further, clones share the same adjacencies as their parents.

The \emph{dual} of a tournament reverses the orientation of each of its directed edges. The analogously defined {Dual Iterated Local Transitivity for Tournaments (or ILTT$_d$)} shares all arcs as defined in the ILTT model, but replaces the subtournament of clones by its dual: if $(x,y)$ is an arc in $G_t,$ then $(y',x')$ is an arc in $G_{t+1}.$ To better distinguish the models, we use the notation $\overrightarrow{G_t}$ and $\overleftarrow{G_t}$ for an ILTT and ILTT$_d$ tournament at time-step $t\ge 1$, respectively. See Figure~\ref{fig:3_cycle_ILTT}.

\begin{figure}[htpb!]
\centering
\includegraphics[scale=0.45]{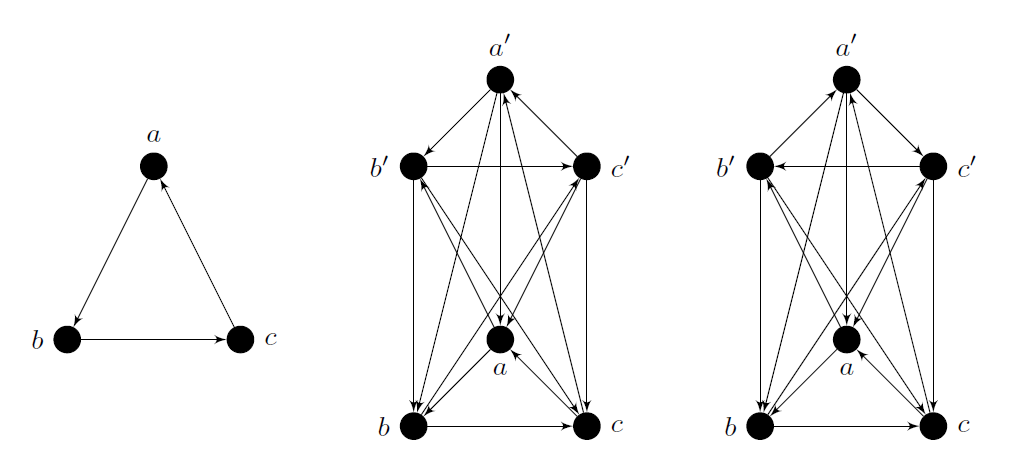}\
\caption{If $G_0$ is the directed 3-cycle, then the first-time step of the ILTT tournament is in the middle and the ILLT$_d$ tournament is on the right.}\label{fig:3_cycle_ILTT}
\end{figure}

We will demonstrate that the ILTT and ILTT$_d$ models simulate properties observed in complex
graphs and digraphs, such as the small-world property. Section~2 considers distances in ILTT and ILTT$_d$ tournaments and shows
that their average distances are bounded above by a constant. We consider motifs in Section~3 and show that while the family of ILTT$_d$ tournaments are universal in the sense that they eventually contain any given tournament (as proved in Theorem~\ref{universality}), the same does not hold for ILTT tournaments. Section~4 focuses on graph-theoretic properties of the models, such as Hamiltonian and spectral properties, and their domination number. We conclude with future directions and open problems.

Throughout the paper, we consider finite, simple tournaments. A tournament is \emph{strong} if for each pair of nodes $x$ and $y$, there are directed paths connecting $x$ to $y$ and $y$ to $x$. If $x_1,x_2,\dots ,x_n$ are nodes of a tournament $G$ so that $(x_i,x_j)$ is an arc whenever $i<j$, then $G$ is a \emph{linear order}. A linear order on three nodes is a \emph{transitive 3-cycle}. For background on tournaments, the reader is directed to \cite{bang-jensen-gutin-2009} and to \cite{west} for background on graphs.  For background on social and complex networks, see \cite{bbook,CL}.

\section{Small world property}

Tournaments are highly dense structures, so local clustering or densification properties are less interesting. However, the presence of short distances remains relevant. In a strong tournament $G$, the distance from nodes $x$ to $y$ is the length of the shortest path from $x$ to $y$, denoted $d_G(x,y)$, or $d(x,y)$ if $G$ is clear from context. Although $d_G$ satisfies the triangle inequality, it is not a metric as it need not be symmetric; that is, $d_G$ is a quasimetric. The maximum distance among pairs of nodes in $G$ is its \emph{diameter}, written $\mathrm{diam}(G)$.

We consider the diameters of ILTT and ILTT$_d$ tournaments by studying how diam($\overrightarrow{G_t}$) and diam($\overleftarrow{G_t}$) relate to the diameter of their bases. As we only consider distances in strong tournaments, we focus on bases that are non-trivial strong tournaments. We first need the following theorem to ensure tournaments generated by the model are strong.

\begin{theorem}\label{ILTT+_strong}
If $G_0$ is a tournament of order at least 3, then for all integers $t\ge 1$, $\overrightarrow{G_t}$ is strong if and only if $G$ is strong.
\end{theorem}

\begin{proof}
We prove the forward direction in induction on $t \ge 0.$ The base case is immediate, so suppose that for a fixed $t\ge 1$, $\overrightarrow{G_t}$ is strong.

The subtournament of $\overrightarrow{G_{t+1}}$ induced by the clones is strong. Let $u$ and $v$ be distinct nodes of $\overrightarrow{G_t}$. We show that $u$ and $v'$ are in the same strong component. Let $P$ be a directed path from $u$ to $v$ in $\overrightarrow{G_t}$, and let $(w,v)$ be the final arc of $P$. Let $P_1$ be the directed path obtained by replacing $(w,v)$ with $(w,v')$ in $P$ and let $P_2$ be a directed path from $v'$ to $u'$.
\begin{figure}[htpb!]
 \centering
\includegraphics[scale=0.45]{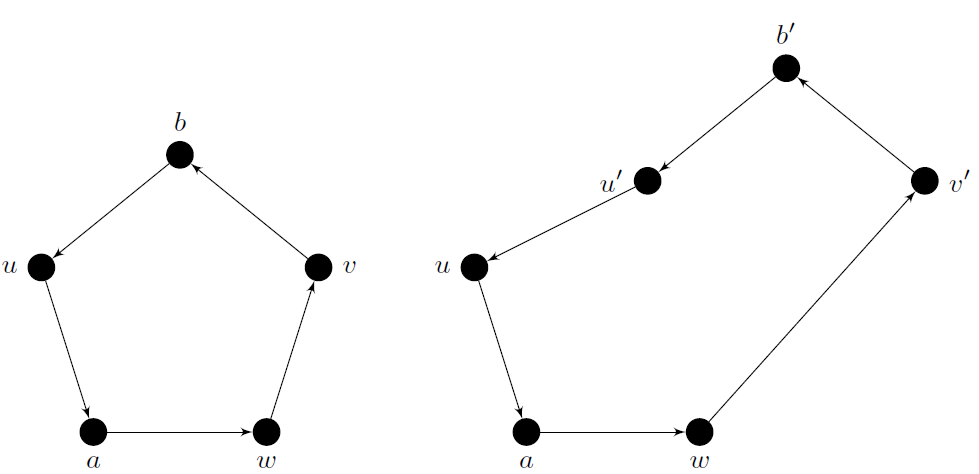}\label{fig:ILTT_G_strong}
\caption{A cycle passing through $u$ and $v$ in $G$ can be used to obtain a directed cycle passing through $u$ and $v'$.}

\end{figure}

If we traverse along $P_1$ from $u$ to $v',$ traverse $P_2$, and then follow the arc $(u',u)$, we have a directed cycle containing $u$ and $v'$. See Figure~2 for an example.

As $\overrightarrow{G_t}$ is strong, $x$ is in a directed 3-cycle $C$. Suppose that $C$ has nodes $x,$ $u,$ and $v$ and arcs $(x,u),$ $(u,v)$, and $(v,x).$ We then have that $(x,u'),$ $(u',v'),$ $(v,x')$, and $(x',x)$ is a directed cycle in $\overrightarrow{G_1}$. It follows that $x$ and $x'$ are in the same strong component. As an aside, we note that there can be no directed 2-path from $x$ to $x'$, and so $d(x,x')=3$. Hence, each pair of nodes of $\overrightarrow{G_{t+1}}$ are in the same strong component, and the forward direction follows.

If $\overrightarrow{G_t}$ is not strong, then there are nodes $x$ and $y$ of $\overrightarrow{G_t}$ so there is no directed path from $x$ to $y.$ If there was a directed path $P$ from $x$ to $y$ in $\overrightarrow{G_{t+1}},$ then by exchanging each clone for its parent in $P$, we would find a directed path from $x$ to $y$ in $\overrightarrow{G_t},$ which is a contraction.  Hence, $\overrightarrow{G_{t+1}}$ is not strong.
\qed
\end{proof}

We consider an analogous result on strong ILTT$_d$ tournaments.
\begin{theorem}
Suppose that $G=G_0$ has no sink. We then have that for all $t\ge 1$, $\overleftarrow{G_t}$ is strong.
\end{theorem}

\begin{proof}
We prove the case that $G_1$ is strong, with the cases $t\ge 2$ following in an analogous way by induction.
For a node $x$ in $G$, there is some $y$ such that $(x,y)$ is an arc. The directed cycle with nodes $(x,y'),(y',x'),(x',x)$ shows that $x$ and $x'$ are in the same strong component.

Suppose $w_1$ and $w_2$ are two nodes in $G$. Without loss of generality, suppose that $(w_1,w_2)$ is an arc. As $G$ has no sink, there is some node $u$ such that $(w_2,u)$ is an arc. The arcs $$(w_1,w_2),(w_2,u'),(u',{w_2}'),({w_2}',{w_1}'),({w_1}',w_1)$$ form a directed cycle containing $w_1$, $w_2,$ ${w_1}'$, and ${w_2}'$. In particular, any two nodes in $G$, any two clones, and any node of $G$ along with a clone are in the same strong component. Hence, $\overleftarrow{G_1}$ is strong. \qed
\end{proof}

Observe that if $G_0$ has a sink, then $\overleftarrow{G_1}$ may not be strong, as is the case for $G_0$ equaling a single directed edge.

We next prove the following lemma on distances in ILTT.

\begin{lemma}\label{diam_L(G)}
Suppose that $G_0$ is strong with order at least 3. For all integers $t\ge 1$ and for distinct nodes $x$ and $y$ of $\overrightarrow{G_{t+1}}$,
    $$d_{\overrightarrow{G_t}}(x,y) = d_{\overrightarrow{G_{t+1}}}(x,y) = d_{\overrightarrow{G_{t+1}}}(x,y') = d_{\overrightarrow{G_{t+1}}}(x',y) = d_{\overrightarrow{G_{t+1}}}(x',y').$$
\end{lemma}

\begin{proof}
Let $x$ and $y$ be nodes of $\overrightarrow{G_{t}}$. If $\alpha \in \{x,x'\}$ and $\beta \in \{y,y'\}$, then for any directed path $P$ from $x$ to $y$, replacing $x$ with $\alpha$ in the first arc of $P$ and replacing $y$ with $\beta$ in the final arc of $P$ is a directed path from $\alpha$ to $\beta$ in $\overrightarrow{G_{t+1}}$. Thus, we have that $d_{\overrightarrow{G_{t+1}}}(\alpha,\beta) \leq d_{\overrightarrow{G_{t}}}(x,y)$.

Let $f:V(\overrightarrow{G_{t+1}}) \rightarrow V(\overrightarrow{G_{t}})$ be the function mapping each clone to its parent, and fixing each parent. We then have that $f$ is a homomorphism from $\overrightarrow{G_{t+1}}$ to the tournament obtained from $\overrightarrow{G_{t}}$ by attaching a loop to each node. Let $\alpha \in \{x,x'\}$ and $\beta \in \{y,y'\}$ as in the first paragraph of the proof. Omitting loops from the image of any path from $\alpha$ to $\beta$ under $f$, we obtain a directed walk from $x$ to $y$ in $G$. Thus, $d_{\overrightarrow{G_{t}}}(x,y) \leq d_{\overrightarrow{G_{t+1}}}(\alpha,\beta)$. The proof follows. \qed
\end{proof}

In an ILTT tournament with a strong base, as noted in the proof of Theorem~\ref{ILTT+_strong}, $d(x,x')= 3.$ Also, note that by the definition of the model, $d(x',x)=1$. Hence, we have the following corollary from these observations and by Lemma~\ref{diam_L(G)}.

\begin{corollary}\label{cord}
If $G=G_0$ is a strong tournament of order at least 3, then an ILTT tournament $G_t$ satisfies $\mathrm{diam}(\overrightarrow{G_t})\leq \max\{\mathrm{diam}(G_0),3\}.$
\end{corollary}

It is straightforward to see that $d(x,x')=2$ in $\overleftarrow{G_t}$. By a result analogous to Lemma~\ref{diam_L(G)} for distances in the ILTT$_d$ model (with proof omitted), we may also derive that if $G_0$ is strong and order at least 3, then $\mathrm{diam}(\overleftarrow{G_t}) \leq  \mathrm{diam}(G_0).$

We next consider the average distance of tournaments from the ILTT and ILTT$_d$ models. The \emph{Wiener index} of a strong tournament $G$ of order $n$ is
$$W(G) = \sum\limits_{(u,v) \in V(G){\times}V(G)} d_G(u,v).$$
The \emph{average distance} of $G$ is
$$L(G) = \frac{W(G)}{n(n-1)}.$$
Note that we do not have to divide by $n \choose 2$ as distances are not necessarily symmetric.

The following formula for $W(\overrightarrow{G_t})$ follows from Lemma~\ref{diam_L(G)} and the definitions.

\begin{theorem}\label{Weiner_R(G)}
Suppose that $G_0$ is strong of order $n \ge 3$. If $t\ge 1$ is an integer, then $$W(\overrightarrow{G_{t+1}}) = 4(2^tn + W(\overrightarrow{G_t})).$$
\end{theorem}

Theorem~\ref{Weiner_R(G)} allows us to prove the following corollary, which shows that ILTT tournaments satisfy the small world property.
Note that ILTT tournaments tend not to deviate much in average length from that of the base when the base has large order.

\begin{corollary}\label{avg_len_R(G)}
    Let $G_0$ be a strong tournament of order $n \ge 3$. We have the following properties.
    \begin{enumerate}
        \item $W(\overrightarrow{G_t}) = 2^{t+1}(2^{t} - 1)n + 4^{t}W(G_0)$.
        \item $L(\overrightarrow{G_t}) \sim \frac{2+(n-1)L(G_0)}{n}$.
    \end{enumerate}
\end{corollary}

\begin{proof}
The proof follows by induction on $t\ge 0$ with the base case immediate.
For the inductive step, by Theorem~\ref{Weiner_R(G)} and the inductive hypothesis, we derive that
    \begin{eqnarray*}
    W(\overrightarrow{G_{t+1}}) &=& 4\cdot 2^{t}n + 4W(\overrightarrow{G_t})\\
    &=& n \left(\sum\limits_{k=0}^t2^k 4^{t+1-k}\right) + 4^{t+1}W(G_0).
    \end{eqnarray*}
As the geometric series $\sum\limits_{k=0}^t2^k4^{t+1-k}$ equals $2^{t+2} (2^{t+1}-1),$ item (1) follows.

For item (2), we have by item (1) that
    \begin{eqnarray*}
    L(G_t) &=& \frac{W(\overrightarrow{G_t)}}{2^t n(2^tn - 1)}\\
    &=&\frac{2^{t+1}(2^t - 1)n}{2^tn(2^tn-1)} + \frac{4^tW(G_0)}{2^t n(2^t n - 1)}\\
    &=&2\frac{2^t-1}{2^t n-1} + (n-1)L(G_0)\frac{2^t}{2^tn - 1}\\
    &\sim&\frac{2}{n} + \frac{(n-1)L(G_0)}{n}.
    \end{eqnarray*}
The proof follows.\qed
\end{proof}

Interestingly, the average distances in ILTT$_d$ tournaments approach the same limit independent of the base. The proof of the following lemma is omitted for space considerations.

\begin{lemma}\label{lemnn}
Let $G=G_0$ be a strong tournament on $n \ge 3$ nodes, and let $\alpha$ be the number of arcs in $G$ not on a directed 3-cycle. For all $t\ge 1,$ we then have that $$W(\overleftarrow{G_t}) = 12{2^{t-1}n\choose2}+\alpha+3\cdot 2^{t-1}n.$$
\end{lemma}

The following corollary gives the average distance for ILTT$_d$ tournaments.

\begin{corollary}
    Let $G_0$ be a strong tournament of order at least 3. We then have that
    $$\lim\limits_{t\rightarrow{\infty}}L(\overleftarrow{G_t}) = \frac{3}{2}.$$
\end{corollary}

\begin{proof}
    Let $\alpha$ be the number of arcs not on a directed 3-cycle in $G_0$. By Lemma~\ref{lemnn} we have that for $t\ge 1$
    \begin{eqnarray*}
    L(\overleftarrow{G_t}) &=& \frac{W(\overleftarrow{G_t})}{2^{t}n(2^{t}n-1)}\\
    &=&\frac{6(2^{t-1}n(2^{t-1}n-1))+\alpha+3(2^{t-1}n)}{2^{t}n(2^{t}n-1)}\\
    &\sim&\frac{3}{2}.
    \end{eqnarray*}
 The proof follows.  \qed
\end{proof}

\section{Motifs and universality}

Motifs are small subgraphs that are important in complex networks as one measure of similarity. For example, the counts of 3- and 4-node subgraphs give a similarity measure for distinct graphs; see \cite{plos} for implementations of this approach using machine learning. In the present section, we give results on subtournaments in ILTT and ILTT$_d$.

The following theorem shows that every linear order is a tournament of an ILTT tournament.

\begin{theorem}\label{linord_occurs}
For an integer $t\ge 1$, the linear order of order $t$ is a subtournament of $\overrightarrow{G_t}$ and $\overleftarrow{G_t}$.
\end{theorem}

\begin{proof}
We give the proof for the ILTT model as the argument is analogous to the ILTT$_d$ model. The proof follows by induction on $t\ge 1$. Suppose \{$x_1,x_2,\ldots,x_{t}$\} is the node set of a linear order in $G_t$. Let $x_0 = x'_1$ in $G_{t+1}$. We then have that the subtournament on $\{x_0,x_1,x_2, \dots ,x_t\}$ is a linear order with $t+1$ nodes in $G_{t+1}.$ \qed
\end{proof}

Despite Theorem~\ref{linord_occurs}, the subtournaments of ILTT tournaments are restricted by their base tournament. As an application of Theorem~\ref{ILTT+_strong} (or seen directly), note that if we begin with a transitive 3-cycle as $G_0$, then for all $t\ge 0$, $\overrightarrow{G_t}$ never contains a directed 3-cycle.

Given a tournament on some finite set of nodes, one can obtain any other tournament on those nodes by simply reversing the orientations of certain arcs. Given two tournaments $G$ and $H$ on the same set of nodes and $m\ge 1$ an integer, we say that $G$ and $H$ \emph{differ by} $m$ arcs if there are $m$ distinct pairs of nodes $x$ and $y$ such that $(x,y)$ is an arc of $G$ while $(y,x)$ is an arc of $H$. We have the following lemma.

\begin{lemma}\label{singlearc_diff}
Let $G_0$ be a tournament. For $t\ge 0$, if $H$ is a subtournament of $\overleftarrow{G_t}$ and $H$ and a tournament $J$ differ by one arc, then $J$ is isomorphic to a subtournament $\overleftarrow{G_{t+1}}$.
\end{lemma}
\begin{proof}
Let $(u,v)$ be an arc of $H$ such that $(v,u)$ is an arc of $J$. We replace $u$ by $u'$ and $v$ by $v'$ in $\overleftarrow{G_{t+1}}$. The subtournament on $(V(H)\setminus \{u,v \}) \cup \{u',v'\}$ is an isomorphic copy of $J$ in $\overleftarrow{G_{t+1}}$. See Figure~\ref{fig:differ_by_one}.
\begin{figure}[htpb!]
\centering
\includegraphics[scale=0.45]{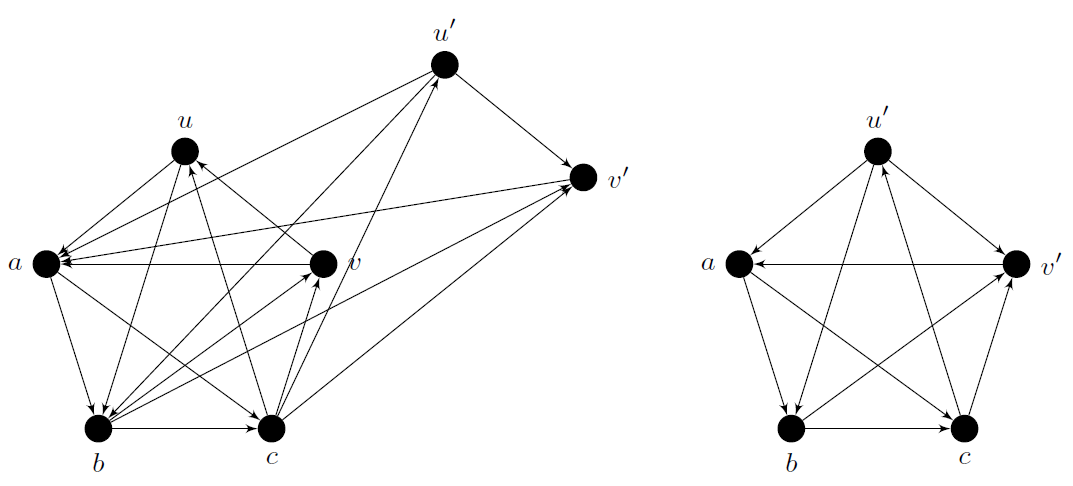}
    \caption{Constructing an isomorphic copy of $J$, which results by reversing the arc $(v,u).$}
    \label{fig:differ_by_one}
\end{figure}
\qed
\end{proof}

Inductively extending Lemma~\ref{singlearc_diff} to tournaments differing by more than one arc and combining this fact with Theorem \ref{linord_occurs}, we find that \emph{any} tournament is subtournament of an ILTT$_d$ tournament with an arbitrary base. We say that the family of ILTT$_d$ tournaments are \emph{universal}.

\begin{theorem}\label{universality}
For each $n\ge 2$, let $\kappa_{n}=n + {n\choose2}$. A tournament $J$ of order $n$ nodes is isomorphic to a subtournament of some $\overleftarrow{G_{r}},$ where $r \le \kappa_n$.
\end{theorem}
\begin{proof}
By Theorem~\ref{linord_occurs}, we may find an isomorphic copy of the linear order $L_n$ with $n$ nodes in $\overleftarrow{G_{n}}.$ We may now iteratively apply Lemma~\ref{singlearc_diff} to $L_n,$ reversing arrows as needed until we arrive at an isomorphic copy of $J$ in some $\overleftarrow{G_{r}}.$ As $J$ and $L_n$ differ by at most ${n\choose2}$ arcs, we have that $r \le \kappa_{n}$. \qed
\end{proof}

\section{Graph-theoretic properties of the models}

This section presents various graph-theoretic results on the ILTT and ILTT$_d$ models. We consider the Hamilitonicity, spectral properties, and the domination number of tournaments generated by the model.

\subsection{Hamiltonicity}
A \emph{Hamiltonian cycle} in a tournament $G$ is a directed cycle that visits each node of $G$ exactly once. The tournament $G$ is said to be $Hamiltonian$ if it has a Hamilton cycle. Suppose $C_1$ is a Hamilton cycle in a tournament $G$ of order $r>3$ passing through nodes $a_1,a_2,\ldots,a_r,$ and $a_1$ in that order. Consider the cycle $C_2$ passing through the nodes $a_k,$ $a_{k+1}',$ and $a_{k+1}$ in that order for $k$ ranging over $\{1,2,\ldots,r\}$ and subscripts taken modulo $r$.
\begin{figure}[htpb!]
\centering
\includegraphics[scale=0.45]{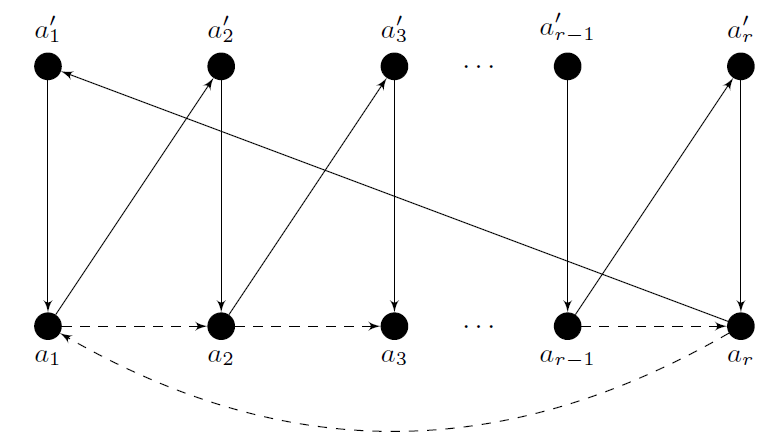}
    \caption{The dashed arrows depict a Hamilton cycle in $G_0$ that is lifted to a Hamilton cycle, depicted by the solid arrows.}
    \label{ILTT+-_Hamilton}
\end{figure}
In both $\overrightarrow{G_1}$ and $\overleftarrow{G_1}$, $C_2$ is a Hamilton cycle. In particular, in $C_2$ each $a_k$ has a unique in-neighbor $a_{k-1}'$ and out-neighbor $a_{k+1}$, and each $a_k'$ has in-neighbor $a_{k-1}$ and out-neighbor $a_k$. Therefore, the following theorem follows by induction on $t\ge 1.$

\begin{theorem}\label{Hamilton_preserved}
If $G_0$ is a Hamiltonian tournament of order at least three, then for all $t\ge 1,$ so are $\overrightarrow{G_t}$ and $\overleftarrow{G_t}$.
\end{theorem}

Camion~\cite{cam} proved that if $G$ is a tournament of order at least three, then $G$ is Hamiltonian if and only if $G$ is strong. We, therefore, have the following corollary of Theorem~\ref{Hamilton_preserved}.

\begin{corollary}
Suppose that the base tournament $G_0$ is order at least 3.
\begin{enumerate}
\item An ILTT tournament is Hamiltonian if and only if its base is strong.
\item An ILTT$_d$ tournament of order is Hamiltonian if its base has no sink.
\end{enumerate}
\end{corollary}

\subsection{Spectral properties}

\newcommand{\mat}[4]{\left(\begin{array}{cc}
        #1 & #2\\
        #3 & #4\\
        \end{array}\right)}

\newcommand{\vect}[2]{\left(\begin{array}{cc}
        #1\\
        #2
        \end{array}\right)}

We next consider eigenvalues of the adjacency matrices of ILTT tournaments. Spectral graph theory is a well-developed area for undirected graphs (see \cite{sgt}) but less so for directed graphs (where the eigenvalues may be complex numbers with non-zero imaginary parts). The following theorem characterizes the eigenvalues of ILTT tournaments.

\begin{figure}[htpb!]
\centering
\includegraphics[scale=0.27]{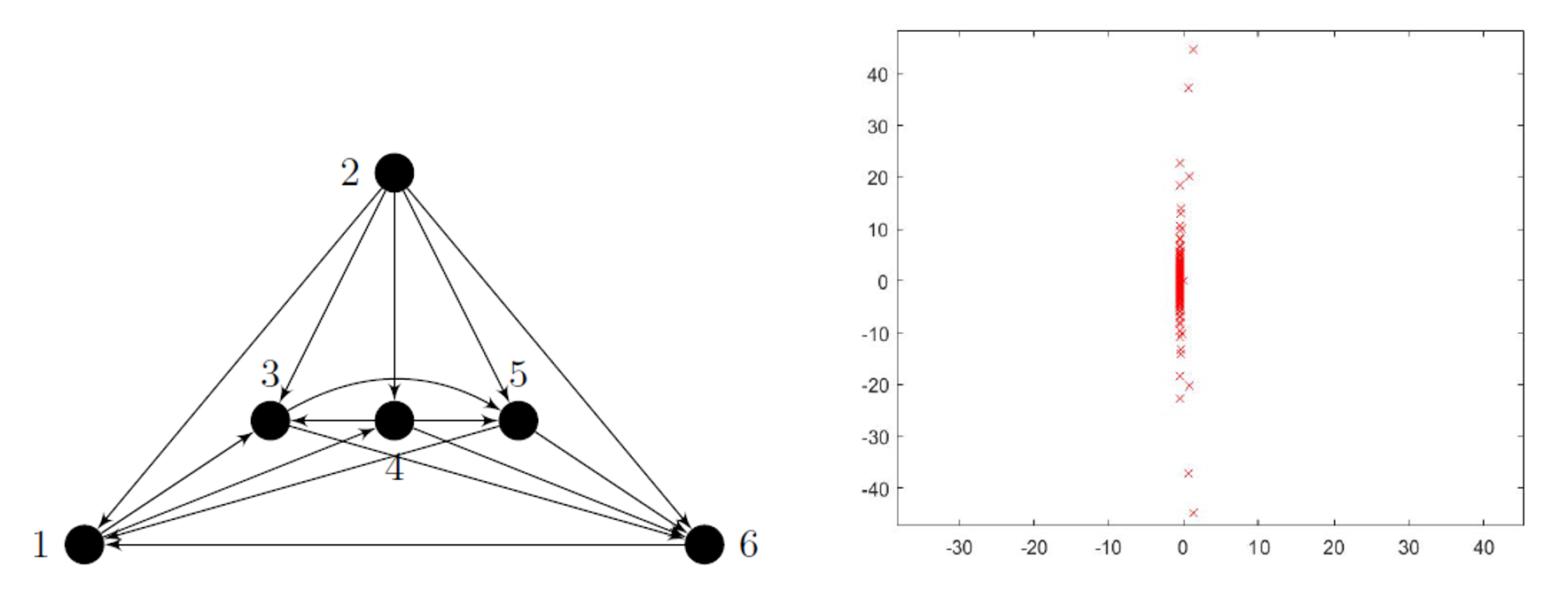}
    \caption{If $G_0$ is the tournament depicted on the left, then the right figure depicts the eigenvalues of the seventh time-step of the ILTT tournament in the complex plane.
    }\label{tttt}
\end{figure}

Theorem~\ref{eig} below gives a characterization of all non-zero eigenvalues of the tournaments arising from the ILTT model in terms of the non-zero eigenvalues of the base. Figure~\ref{tttt} gives a visualization of the eigenvalues in an example of an ILTT tournament.

\begin{theorem}\label{eig}
    If $A_t$ is the adjacency matrix of an ILTT tournament $\overrightarrow{G_t}$ of order $n$, then the following items hold.
    \begin{enumerate}
        \item The value $-1/2$ is not an eigenvalue of $A_{t+1}$.
        \item For each non-zero eigenvalue $\lambda$ of $A_t$, $\mu = \pm (\lambda^2 + \lambda)^{1/2} + \lambda$ is an eigenvalue of $A_{t+1}$ and these are all the non-zero eigenvalues of $A_{t+1}$.
    \end{enumerate}
\end{theorem}
\begin{proof}
    Suppose $\mu \neq 0$ is an eigenvalue of $A_{t+1}$ and let $\mathbf{v},\mathbf{w}\in\mathbb{C}^n$ not both the zero vector $\mathbf{0}$ be such that $\mu\vect{\mathbf{v}}{\mathbf{w}} = A_{t+1}\vect{\mathbf{v}}{\mathbf{w}}$. We then have that $A_{t+1} = \mat{A_t}{A_t}{I+A_t}{A_t}$. Hence,
    \begin{eqnarray}
        \mu \mathbf{v} &=& A_t\mathbf{v} + A_{t}\mathbf{w}\\
        \mu \mathbf{w} &=& \mathbf{v} + A_t\mathbf{v} + A_t\mathbf{w}.
    \end{eqnarray}
    We then have that $\mu \mathbf{w} = (1 + \mu)\mathbf{v}$. By (1), it follows that $\frac{1+\mu}{\mu}A_t\mathbf{v} + A_t\mathbf{v} = \mu\mathbf{v}$, which gives that $$\frac{1+2\mu}{\mu}A_t\mathbf{v} = \mu \mathbf{v}.$$

    For (1), note that, by (1) and (2), if $\mu = -1/2$, then
    \begin{eqnarray*}
        A_t\mathbf{v} + A_t\mathbf{w} &=& -\frac{1}{2}\mathbf{v}\\
        \mathbf{v} + A_t\mathbf{v} + A_t\mathbf{w} &=& -\frac{1}{2}\mathbf{w}.
    \end{eqnarray*}
    Thus, $\mathbf{v} = -\mathbf{w}$ and $\mathbf{v} = \mathbf{0}$, which is a contradiction. Therefore, $\mu\neq{-1/2}$.

    For (2), note that the computation in the proof of part (1) also shows that $$A_t\mathbf{v} = \frac{\mu^2}{1+2\mu}\mathbf{v},$$ so that $\lambda = \frac{\mu^2}{1+2\mu}$ is an eigenvalue of $A_t$. In particular, $\lambda$ satisfies $\mu^2 - 2\mu\lambda -\lambda = 0$ or $(\mu-\lambda)^2 - \lambda - \lambda^2 = 0$ or $\mu = \pm (\lambda^2+\lambda)^{1/2}+\lambda$.
    Conversely, suppose $\lambda \neq 0$ is an eigenvalue of $A_t$ and $\mathbf{v}\in\mathbb{C}^n$ such that $A_t\mathbf{v} = \lambda{\mathbf{v}}$. Let $\mu$ be a root of the quadratic $x^2 - 2\lambda{x} - \lambda$ and let $\mathbf{w} = \frac{1+\mu}{\mu}\mathbf{v}$. We then have that $\mu\neq0$ and
    \begin{eqnarray*}
       A_{t+1}\vect{\mathbf{v}}{\mathbf{w}} &=& \vect{A_t\mathbf{v}+ \frac{1+\mu}{\mu}A_t\mathbf{v}}{\mathbf{v} + A_t\mathbf{v} + \frac{1+\mu}{\mu}A_t\mathbf{v}}\\
        &=& \vect{\lambda \mathbf{v} + \frac{1+\mu}{\mu}\lambda \mathbf{v}}{(1+\lambda)\mathbf{v} + \frac{1+\mu}{\mu}\lambda \mathbf{v}}\\
        &=&\vect{(2\lambda + \frac{\lambda}{\mu})\mathbf{v}}{(1 + 2\lambda + \frac{\lambda}{\mu})\mathbf{v}}.
    \end{eqnarray*} The equation $\mu^2 -2\lambda\mu -\lambda = 0$ implies that $\frac{\lambda}{\mu} = \mu - 2\lambda$. Hence,
    $$    A_{t+1}\vect{\mathbf{v}}{\mathbf{w}} = \vect{\mu \mathbf{v}}{(1+\mu)\mathbf{v}} = \mu\vect{\mathbf{v}}{\mathbf{w}},$$ so that $\mu$ is an eigenvalue of $A_{t+1}$. \qed
\end{proof}

\subsection{Domination numbers}

Dominating sets contain nodes with strong influence over the rest of the network. Several earlier studies focused on the domination number of complex network models; see \cite{domb,cooper}. In a tournament $G$, an \emph{in-dominating set} $S$ has the property that for each $u$ not in $S$, there is a $v \in S$ such that $(u,v)$. The dual notion is an \emph{out-dominating set}. The cardinality of a minimum order in-dominating set is the \emph{in-domination number} of $G$, written $\gamma^-(G);$ the dual notion is the \emph{out-domination number} of $G$, written $\gamma^+(G).$

In this final subsection, we investigate the domination number of ILTT tournaments. Proofs are omitted owing to space considerations. We first need a lemma relating in- and out-dominating sets of the tournaments $\overrightarrow{G_t}$ and $\overleftarrow{G_t}$.

\begin{lemma}\label{domsets}
    Let $G_0$ be a tournament.
    \begin{enumerate}
        \item If $S$ is an in- or out-dominating set for $\overrightarrow{G_t}$ for $t\ge 1$, then $$(S \cap V(G))\cup\{v \in V(G): v' \in T\}$$ is the same for $G_0$.
        \item If $S$ is an in-dominating set for $G_0$, then $\{v':v\in S\}$ is an in-dominating set for $\overrightarrow{G_t}$.
        \item If $S$ is an out-dominating set for $G_0$, then $S$ is an out-dominating set for $\overrightarrow{G_t}$.
        \item If $S$ is a minimal in- or out-dominating set for $\overrightarrow{G_t}$, then $S$ cannot contain both a node in $V(G_0)$ and its clone.
        \item If $S$ is an in- or out-dominating set for $\overleftarrow{G_t}$, then $$(S\cap V(G_0))\cup\{v\in V(G_0):v'\in S\}$$ is the same for $G_0$.
    \end{enumerate}
\end{lemma}

Lemma~\ref{domsets} now gives the following theorem, which describes the evolution of the domination numbers of ILTT and ILTT$_d$ tournaments.

\begin{theorem}\label{tdom}
    If $G_0$ is a tournament, then the following identities hold for all $t\ge 1$:
    \begin{enumerate}
        \item $\gamma^\pm(G_0) = \gamma^\pm(\overrightarrow{G_t})$.
        \item $\gamma^+(G_0) = \gamma^+(\overleftarrow{G_t})$.
        \item $\gamma^-(G_0) \leq \gamma^-(\overleftarrow{G_t})$.
    \end{enumerate}
\end{theorem}
Note that a consequence of Theorem~\ref{tdom} is that the domination numbers of ILTT tournaments equal the domination numbers of their base. An analogous fact holds for the ILTT$_d$ model in the case of the out-domination number.

\section{Conclusion and further directions}

We considered new models ILTT and ILTT$_d$ for evolving, complex tournaments based on local transitivity. We proved that both models generate tournaments with small distances and that ILTT$_d$ tournaments contain all tournaments as subtournaments. We investigated when tournaments from the models were strong or Hamiltonian, and considered their in- and out-domination numbers. We also considered spectral properties of ILTT tournaments.

Studying the counts of various motifs in the models, such as directed and transitive cycles, would be interesting. Insights into the motifs found in the model may shed light on the resulting tournament limits (which are an extension of graphons to tournaments; see \cite{TL}). We did not consider spectral properties of ILTT$_d$ tournaments, and the cop number and automorphism group of tournaments generated by the model would be interesting to analyze. We may also consider randomized versions of the models, where in each time-step, either the subtournament of clones is either the existing tournament or its dual, chosen with a given probability. It would also be interesting to consider the effects of placing a random tournament on the subtournament of clones, where arcs are oriented with a given probability.

\end{document}